\documentclass[12pt]{article}

\usepackage{amsmath,amsfonts,amssymb,color,epsfig,graphics,graphicx,latexsym,revsymb,theorem,url,verbatim}

\usepackage{fullpage}

\newtheorem{definition}{Definition}
\newtheorem{proposition}[definition]{Proposition}
\newtheorem{lemma}[definition]{Lemma}

\newtheorem{theorem}[definition]{Theorem}
\newtheorem{corollary}[definition]{Corollary}
\newtheorem{conjecture}[definition]{Conjecture}

\newtheorem{remark}[definition]{Remark}
\newtheorem{example}[definition]{Example}
\newtheorem{question}[definition]{Question}

\def\squareforqed{\hbox{\rlap{$\sqcap$}$\sqcup$}}
\def\qed{\ifmmode\squareforqed\else{\unskip\nobreak\hfil
\penalty50\hskip1em\null\nobreak\hfil\squareforqed
\parfillskip=0pt\finalhyphendemerits=0\endgraf}\fi}
\def\endenv{\ifmmode\;\else{\unskip\nobreak\hfil
\penalty50\hskip1em\null\nobreak\hfil\;
\parfillskip=0pt\finalhyphendemerits=0\endgraf}\fi}
\newenvironment{proof}{\noindent \textbf{{Proof.~} }}{\qed}
\def\Dbar{\leavevmode\lower.6ex\hbox to 0pt
{\hskip-.23ex\accent"16\hss}D}
\makeatletter
\def\url@leostyle{%
  \@ifundefined{selectfont}{\def\UrlFont{\sf}}{\def\UrlFont{\small\ttfamily}}}
\makeatother
\def\bcj{\begin{conjecture}}
\def\ecj{\end{conjecture}}
\def\bcr{\begin{corollary}}
\def\ecr{\end{corollary}}
\def\bd{\begin{definition}}
\def\ed{\end{definition}}
\def\bea{\begin{eqnarray}}
\def\eea{\end{eqnarray}}
\def\bem{\begin{enumerate}}
\def\eem{\end{enumerate}}
\def\bex{\begin{example}}
\def\eex{\end{example}}
\def\bim{\begin{itemize}}
\def\eim{\end{itemize}}
\def\bl{\begin{lemma}}
\def\el{\end{lemma}}
\def\bpf{\begin{proof}}
\def\epf{\end{proof}}
\def\bpp{\begin{proposition}}
\def\epp{\end{proposition}}
\def\bqu{\begin{question}}
\def\equ{\end{question}}
\def\br{\begin{remark}}
\def\er{\end{remark}}
\def\bt{\begin{theorem}}
\def\et{\end{theorem}}

\def\btb{\begin{tabular}}
\def\etb{\end{tabular}}

\newcommand{\nc}{\newcommand}

\def\a{\alpha}
\def\b{\beta}

\def\d{\delta}

\def\r{\rho}
\def\s{\sigma}

\def\ph{\varphi}

\def\ps{\psi}

 \nc{\bA}{{\bf A}} \nc{\bB}{{\bf B}} \nc{\bC}{{\bf C}}
 \nc{\bD}{{\bf D}} \nc{\bE}{{\bf E}} \nc{\bF}{{\bf F}}
 \nc{\bG}{{\bf G}} \nc{\bH}{{\bf H}} \nc{\bI}{{\bf I}}
 \nc{\bJ}{{\bf J}} \nc{\bK}{{\bf K}} \nc{\bL}{{\bf L}}
 \nc{\bM}{{\bf M}} \nc{\bN}{{\bf N}} \nc{\bO}{{\bf O}}
 \nc{\bP}{{\bf P}} \nc{\bQ}{{\bf Q}} \nc{\bR}{{\bf R}}
 \nc{\bS}{{\bf S}} \nc{\bT}{{\bf T}} \nc{\bU}{{\bf U}}
 \nc{\bV}{{\bf V}} \nc{\bW}{{\bf W}} \nc{\bX}{{\bf X}}
 \nc{\bZ}{{\bf Z}}

\nc{\cA}{{\cal A}} \nc{\cB}{{\cal B}} \nc{\cC}{{\cal C}}
\nc{\cD}{{\cal D}} \nc{\cE}{{\cal E}} \nc{\cF}{{\cal F}}
\nc{\cG}{{\cal G}} \nc{\cH}{{\cal H}} \nc{\cI}{{\cal I}}
\nc{\cJ}{{\cal J}} \nc{\cK}{{\cal K}} \nc{\cL}{{\cal L}}
\nc{\cM}{{\cal M}} \nc{\cN}{{\cal N}} \nc{\cO}{{\cal O}}
\nc{\cP}{{\cal P}} \nc{\cQ}{{\cal Q}} \nc{\cR}{{\cal R}}
\nc{\cS}{{\cal S}} \nc{\cT}{{\cal T}} \nc{\cU}{{\cal U}}
\nc{\cV}{{\cal V}} \nc{\cW}{{\cal W}} \nc{\cX}{{\cal X}}
\nc{\cZ}{{\cal Z}}

\nc{\hA}{{\hat{A}}} \nc{\hB}{{\hat{B}}} \nc{\hC}{{\hat{C}}}
\nc{\hD}{{\hat{D}}} \nc{\hE}{{\hat{E}}} \nc{\hF}{{\hat{F}}}
\nc{\hG}{{\hat{G}}} \nc{\hH}{{\hat{H}}} \nc{\hI}{{\hat{I}}}
\nc{\hJ}{{\hat{J}}} \nc{\hK}{{\hat{K}}} \nc{\hL}{{\hat{L}}}
\nc{\hM}{{\hat{M}}} \nc{\hN}{{\hat{N}}} \nc{\hO}{{\hat{O}}}
\nc{\hP}{{\hat{P}}} \nc{\hR}{{\hat{R}}} \nc{\hS}{{\hat{S}}}
\nc{\hT}{{\hat{T}}} \nc{\hU}{{\hat{U}}} \nc{\hV}{{\hat{V}}}
\nc{\hW}{{\hat{W}}} \nc{\hX}{{\hat{X}}} \nc{\hZ}{{\hat{Z}}}

\def\diag{\mathop{\rm diag}}
\def\dim{\mathop{\rm Dim}}

\def\lin{\mathop{\rm span}}

\def\tr{\mathop{\rm Tr}}

\def\GL{{\mbox{\rm GL}}}

\def\Un{{\mbox{\rm U}}}

\def\op{\oplus}
\def\ox{\otimes}

\def\sue{\subseteq}

\def\we{\wedge}
\newcommand{\bra}[1]{\langle#1|}
\newcommand{\ket}[1]{|#1\rangle}
\newcommand{\proj}[1]{| #1\rangle\!\langle #1 |}
\newcommand{\ketbra}[2]{|#1\rangle\!\langle#2|}
\newcommand{\braket}[2]{\langle#1|#2\rangle}
\newcommand{\wetw}[2]{|#1\rangle\wedge|#2\rangle}

\newcommand{\wefo}[4]{|#1\rangle\wedge|#2\rangle\wedge|#3\rangle\wedge|#4\rangle}

\def\Dbar{\leavevmode\lower.6ex\hbox to 0pt
{\hskip-.23ex\accent"16\hss}D}

\begin{document}

\title{Universal Subspaces for Local Unitary Groups of Fermionic Systems}

\author{Lin Chen$^{1,2,3}$, Jianxin Chen$^{2,4}$, Dragomir {\v{Z} \Dbar}okovi{\'c}$^{1,2}$, Bei Zeng$^{2,4}$\\
$1$. Department of Pure Mathematics,\\ University of Waterloo,
Waterloo, Ontario, Canada\\
$2$. Institute for Quantum Computing,\\ University of
Waterloo, Waterloo, Ontario, Canada\\
$3$. Center for Quantum Technologies,\\ National University of
Singapore, Singapore\\
$4$. Department of Mathematics \&
  Statistics,\\ University of Guelph,
  Guelph, Ontario, Canada}

\date{\today}
\maketitle

\begin{abstract}
Let $\mathcal{V}=\wedge^N V$ be the $N$-fermion Hilbert space with $M$-dimensional single particle space $V$ and $2N\le M$. We refer to the unitary group $G$ of $V$ as the local unitary (LU) group. We fix an orthonormal (o.n.) basis
$\ket{v_1},\ldots,\ket{v_M}$ of $V$. Then the Slater determinants $e_{i_1,\ldots,i_N}:=
\ket{v_{i_1}\we v_{i_2}\we\cdots\we v_{i_N}}$ with
$i_1<\cdots<i_N$ form an o.n. basis of $\cV$.
Let $\cS\subseteq\cV$ be the subspace spanned by all
$e_{i_1,\ldots,i_N}$ such that the set $\{i_1,\ldots,i_N\}$
contains no pair $\{2k-1,2k\}$, $k$ an integer.
We say that the $\ket{\psi}\in\cS$ are single occupancy states
(with respect to the basis $\ket{v_1},\ldots,\ket{v_M}$).
We prove that for $N=3$ the subspace $\cS$ is universal, i.e., each $G$-orbit in $\cV$ meets $\cS$, and
that this is false for $N>3$. If $M$ is even, the well known BCS states are not LU-equivalent to any single occupancy state.
Our main result is that for $N=3$ and $M$ even there is a universal subspace $\cW\subseteq\cS$ spanned by $M(M-1)(M-5)/6$ states $e_{i_1,\ldots,i_N}$. Moreover the number $M(M-1)(M-5)/6$ is minimal.
\end{abstract}

\section{Introduction}
\label{sec:Intro}

Entanglement is at the heart of quantum information theory. It makes
possible secure and high rate information transmission, fast
computational solution of certain important problems, and efficient
physical simulation of quantum phenomena~\cite{NC00,Horodecki}. The
most fundamental property for any kind of study of entanglement is
that it is invariant under local unitary (LU) transformations. The
celebrated Schmidt decomposition for bipartite pure states provides
a canonical form for these states under LU, which makes possible a
complete understanding of their entanglement properties. In
particular, the entanglement measure is characterized by the entropy
of Schmidt coefficients.

For multipartite systems, however, no direct generalization is
possible. The study of entanglement for multipartite pure states is
hence much more challenging than the bipartite case. Considerable
efforts have been taken during the past decade. However no complete
satisfactory theory can be reached as for the bipartite
case~\cite{Horodecki,AFOV08}. Still, the first step toward the
understanding of multipartite entanglement is to study the orbits of
pure quantum states under LU. Acin et al started from looking at the
simplest nontrivial case of a three qubit system~\cite{AAC+00,ajt01}. They
introduced the concept of local basis product states (LBPS), which
are three-qubit product states from a fixed set of single particle
orthonormal basis. They studied the property of subspaces spanned by
LBPS such that every three-qubit pure state is LU equivalent to some
states in such a subspace. This kind of subspace can hence be called
`universal subspace'~\cite{AAC+00,ad10}.

Acin et al. showed that five LBPS are enough to span a universal
subspace for three qubits~\cite{AAC+00}. For one such choice given by
$\{\ket{000},\ket{100},\ket{101},\ket{110},\ket{111}\}$, they can
further restrict the regions for the superposition coefficients such
that it gives a canonical form for three qubits~\cite{ad10}. That is, each LU
orbit corresponds to only one set of parameters in the canonical
form. Generalizations to $N>3$ qubits are developed and it is found
that a universal subspace needs to be spanned by $2^N-N$ LBPS, and
such a minimal set of LBPS are identified~\cite{CHS00}.
Unfortunately, this does not provide much saving compared to the
$2^N$-dimensional Hilbert space. This is simply due to the fact the
LU group is such a small group compared to the entire unitary group
$\Un(2^N)$.

Entanglement theory for identical particle systems has also been
considerably developed during the past decade~\cite{AFOV08,
SLM01,SCK+01,LZL+01,PY01,ESB+02}. Entanglement for identical
particle systems cannot be discussed in usual way as for the
distinguishable particle case. This is because the symmetrization
for bosonic systems and antisymmetrization for fermionic systems for
the wave functions may introduce `pseudo entanglement' which is not
accessible in practice. It is now widely agreed that non-entangled
states are reasonably corresponding to the form $\ket{v}^{ N}$ for
bosonic states~\cite{PY01,AFOV08, ESB+02} (there was indeed
suggestion that non-entangled bosonic state corresponds to
$\ket{v_1}\vee\ket{v_2}\vee\cdots\vee\ket{v_N}$~\cite{LZL+01},
however most of the literatures accept $\ket{v}^{ N}$) and to the
form of Slater determinants
$\ket{v_1}\wedge\ket{v_2}\we\cdots\wedge\ket{v_N}$ for fermionic
states~\cite{AFOV08, ESB+02}.

For bipartite fermionic pure states, a direct generalization of
Schmidt decomposition is available~\cite{SCK+01}. That is, every
bipartite fermionic pure state is LU equivalent to a Slater
decomposition $\sum_i\lambda_i \ket{\alpha_i}\wedge\ket{\beta_i}$,
where
$\bra{\alpha_i}\a_j\rangle=\bra{\b_i}\beta_j\rangle=\delta_{ij}$
and $\bra{\alpha_i}\beta_j\rangle=0$.
This thus allows a complete understanding of entanglement properties
for bipartite fermionic pure states. Many interesting analogues to
the bipartite distinguishable particle case have also been
identified, such as concurrence and magic
basis~\cite{SCK+01,ESB+02}. The entanglement of formation (EOF) for
mixed states has also been investigated for the single particle
state Hilbert space of dimension four. In particular, a formula was
identified~\cite{SCK+01,ESB+02} and is similar to the known formula for
two-qubit states by Wootters \cite{wootters98}.

The generalization to more than two fermions has also been
discussed~\cite{ESB+02}. However similar to the distinguishable
particle case, there is no Slater decomposition for multipartite
fermioinic system. Then it is more difficult to study their
entanglement properties. The first step is still using the LU
orbits. Due to particle-hole duality, 
in most cases one will need to have $M\geq 2N$ to obtain meaningful
discussions, where $M$ is the dimension for the single particle state
Hilbert space $V$. For examples with $M >2N$, see the BCS states
defined in Eq. \eqref{ea:BCS}. The reason is that the Pauli
exclusion principle requires $M\geq N$, and the particle-hole
duality gives equivalence between $N,M$ and $M-N, M$ systems. Also,
it is even nontrivial to see whether a given fermionic pure state is
unentangled (i.e. LU equivalent to a single Slater determinant),
where one needs to check the Grassmann-Pl\"{u}cker relations (see
e.g. p. A III.172 Eq. (84-(J,H)) in~\cite{nb1970}, Prop 11-32
in~\cite{bh07}, and~\cite{ESB+02}).

In this paper, we take one further step beyond just the Grassmann-
Pl\"{u}cker relations for the entanglement properties of
multipartite fermionic pure states. We study the universal subspaces
spanned by Slater determinants built from orthonormal local basis.
So we can build a fermionic analogue of the universal subspace
spanned by LBPS in the distinguishable particle case. The obvious
difficulty is the above-mentioned $M\geq 2N$ condition, for which
one can no longer hope for any elementary method (e.g. linear
algebra) as was useful in the $N$-qubit case. However, this
difficulty, once overcome, might become an advantage given the
relatively large LU group $\Un(M)$ compared to the entire unitary
group on the $N$-particle fermion space $\mathcal{V}=\wedge^N V$
which is of dimension $M\choose N$. In other words, one can ideally
hope for a `saving' of the dimension of $\Un(M)$, which is of the
order $M^2$. That is, the best one can hope for the dimension of a
universal subspace is of the order ${M\choose N} -M^2$ for $N>2$.

Our focus will be mainly on the $N=3$ case, where we show that such
a saving of order $M^2$ is indeed achievable. Not only we want to
find a universal subspace which mathematically achieves the order
${M\choose N} -M^2$, but also we want the universal subspace with a
clear physical meaning. To obtain such a universal subspace, we
introduce a configuration for the Slater determinants called single
occupancy. This is typical in condensed matter physics for studying
properties of strongly-correlated electron systems, for instance
high temperature superconducting ground state of
doped Mott insulator (see e.g. ~\cite{Weng11}).

Single occupancy fermionic states are also of interest in quantum
information theory in recent years, as there exists a qubit to
fermion mapping which is one to one between qubit states and single
occupancy fermionic states. This mapping was used to show the QMA
completeness of the fermionic N-representability problem, that is,
the overlapping N-representability problem is hard even with the
existence of a quantum computer~\cite{Liu,ak08}. The same mapping
was also used to disprove a 40-year old conjecture in quantum
chemistry using methods of quantum error-correcting
codes~\cite{OCZ+11}.

The physics considerations lead to the following picture. We
illustrate this in Fig.~\ref{fig:sites}. Suppose we have in total
$K$ sites. At each site we have three possible states: 1) no spin (a
hole, shown as a dot); 2) spin up (shown as arrow up); 3) spin down
(shown as arrow down). Now suppose we have in total $N$ spins. Then
there are in total $2K\choose N$ spin configurations (note that once
the spin states are fixed, the holes are fixed). As an example,
Fig.~\ref{fig:4sites} shows the case $N=3$ and $K=4$, i.e. $3$ spins
in $8$-dimensional space. Figure A is a single occupancy
configuration which has at most $1$ spin on each site. A single
occupancy state is a superposition of single occupancy
configurations. Figure B is not a single occupancy configuration
where site $1$ is occupied by $2$ spins (i.e. double occupancy on
site $1$).  While the single occupancy configuration is defined for
$M=2K$ which is even, one can readily generalize the configuration
to the case of odd $M$, leaving one single-particle spin state
unpaired.

\begin{figure}[htb]
\centering
\includegraphics[scale=0.7]{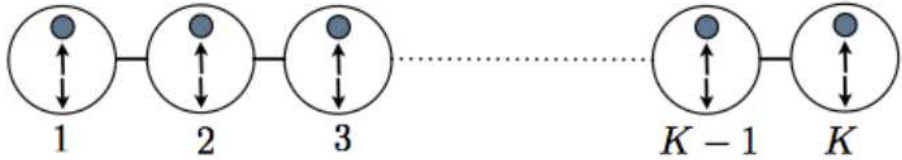}
\caption{Physics picture for single and double occupancy}
\label{fig:sites}
\end{figure}

\begin{figure}[htb]
\centering
\includegraphics[scale=0.6]{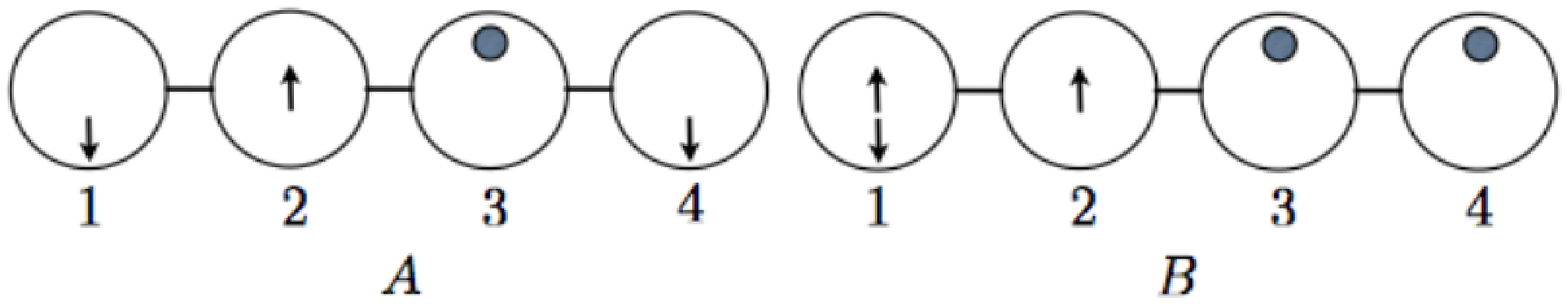}
\caption{4 sites}
\label{fig:4sites}
\end{figure}

To be more precise, for the single particle Hilbert space $V$ of
dimension $M=2K$, which is a tensor product of the spatial part with
dimension $K$ (with a fixed basis
$\ket{g_1},\ket{g_2},\ldots,\ket{g_K}$, which each $\ket{g_j}$
highly localized on the $j$th site) and the spin part with dimension
$2$ (with a fixed basis $\ket{\uparrow},\ket{\downarrow}$), then a
single occupancy state is a Slater determinant where
$\ket{g_j}\otimes\ket{\uparrow}, \ket{g_j}\otimes\ket{\downarrow}$
do not show up at the same time for any $j=1,2,\ldots K$. For
instance, the single occupancy state shown in Fig.~\ref{fig:4sites}A
can then be written as
$(\ket{g_1}\otimes\ket{\downarrow})\wedge(\ket{g_2}\otimes\ket{\uparrow})\wedge(\ket{g_4}
\otimes\ket{\downarrow})$. However, the state given in
Fig.~\ref{fig:4sites}B is
$(\ket{g_1}\otimes\ket{\uparrow})\wedge(\ket{g_1}\otimes\ket{\downarrow})\wedge(\ket{g_2}
\otimes\ket{\uparrow})$, which is not single occupancy as both
$\ket{g_1}\otimes\ket{\uparrow},\ket{g_1}\otimes\ket{\downarrow}$
show up.

For a simplified notation, we can relabel the $2K$ single particle
basis
$\{\ket{g_j}\otimes\ket{\uparrow},\ket{g_j}\otimes\ket{\downarrow}\}$
as $\{\ket{2j-1},\ket{2j}\}$, for $j=1,2\ldots,K$. The $M$ odd case
is hence similar, just that one can correspond the first $M-1$ basis
states, denoted by $\{\ket{2j-1},\ket{2j}\}$ to
$\{\ket{g_j}\otimes\ket{\uparrow},\ket{g_j}\otimes\ket{\downarrow}\}$
for $j=1,2\ldots,(M-1)/2$, and then the last `unpaired' basis state
$\ket{M}$ to $\ket{g_{(M+1)/2}}\otimes\ket{\uparrow}$, as the state
$\ket{g_{(M+1)/2}}\otimes\ket{\downarrow}$ is `inaccessible' by any
of the $N$-fermions. Note that despite the choice of the labelling,
the basis of the Hilbert space is fixed throughout the paper. This
corresponds to a simple fact that the spatial sites are physically
fixed.

Our main theorem of the paper proves that for $N=3$, the single
occupancy subspace $\mathcal{S}$ (i.e. subspace of all single
occupancy states) is universal. In other words, any $N=3$ fermionic
state is LU equivalent to a single occupancy state. Here by LU we
mean that local unitary transformation on the entire $M$-dimensional
single particle Hilbert, i.e. elements in $\Un(M)$, not only unitaries on the spatial part (or
the spin part) of the wavefunction only. This definition of LU
agrees with those used in quantum information community for
fermionic systems (see, e.g. ~\cite{AFOV08, ESB+02}), which preserves
entanglement properties, such as the spectra of reduced density
matrices.

For $M$ even, we can also obtain universal subspaces contained in
$\mathcal{S}$ whose dimension is minimal. Our main tool is to apply
Theorem 4.2 of~\cite{ad10} to the case of LU groups with
representation on the $N$-fermion Hilbert space
$\mathcal{V}=\wedge^N V$. It should be emphasized that this
application in our case is highly non-trivial. Indeed, after
selecting a candidate for a universal subspace, one has to prove
that the main condition of the theorem is satisfied. For small $M$
this can be done by using a computer, but the proof that it can be
applied for all $M$ is hard to find. Finding a universal subspace of
minimum dimension is much harder and we succeeded to find one which
works for even $M$. For the case when $M$ is odd, one may conjecture
that the candidate subspace described in Proposition \ref{pp:BSOV}
(see Eq. \eqref{eq:neparnoM}) is universal but this has been
verified only for $M=7,9,11$.

For the smallest non-trivial case of $N=3,M=6$ where $\dim\wedge^3
V=20$, a direct application of our results gives universal subspaces
with dimension $5$. This is indeed a big saving. It is interesting
to compare this result to Acin et. al's result of three qubits,
where they also have universal subspaces of dimension $5$, which are
spanned by LBPS. This may be related to the fact that, when $M=2N$,
each qubit state corresponds to a single occupancy state based on
the qubit to fermion mapping mentioned above. It should also be
mentioned that, although it is not directly applicable, our work may
shed light on the understanding of the N-representability
problem~\cite{Coleman}, where the single particle eigenvalues are
invariant under LU. Indeed, our universal subspace with dimension
$5$ for the case $N=3,M=6$ gives an alternative proof of the
N-representability equalities and inequalities in that
case~\cite{BD72,Kly05,Beth07,Kly09}.

To complete our study of single occupancy, we further prove that the
single occupancy subspace is not universal for $N>3$, i.e. not all
fermionic states are LU equivalent to a single occupancy state,
though for fixed $N$ almost all states are single occupancy in the
large $M$ limit. Our argument is based on dimension counting. For
concrete examples, we show that for $N$ even, BCS states are not LU
equivalent to any single occupancy state. This is intuitive as BCS
states are always `paired' so they are most unlikely to be
transformed with LU to something unpaired.

We organize our paper as follows. In Sec.~\ref{sec:Pre}, we provide preliminaries used throughout the paper. We formally define antisymmetric tensors, decomposable states (Slater determinants), LU operations, universal subspaces and single occupancy states. We also discuss the lower bound for dimension of universal subspaces. It sets the goal for later work. In Sec.~\ref{sec:USN3}, we consider universal subspaces for $N=3$. We prove our main theorem that for $N=3$, the single occupancy subspace $\mathcal{S}$ is universal, see Theorem \ref{thm:BSOV}. For $M$ even, we can also obtain subspaces of $\mathcal{S}$ whose dimension is minimal, and explicit choice of Slater determinants for those subspaces are
given. In Sec.~\ref{sec:USNgeq3}, we try to generalize our theorem to the $N>3$ case but obtain a negative result. That is, we prove that the single occupancy subspace is not universal for any
$M\geq 2N$. So there always exist some states which are not LU equivalent to single occupancy states. For $N$ even, we use BCS states as concrete examples of states that are not LU equivalent to any single occupancy states. Finally, a brief summary and discussion will be given in Sec.~\ref{sec:Sum}.

\section{Preliminaries}
\label{sec:Pre}

\subsection{Decomposable states and inner product}

Let $V$ be a complex Hilbert space of dimension $M$ and
$\cH=\otimes^N V$ the $N$th tensor power of $V$. The inner
product on $V$ extends to one on $\cH$ such that
$$
\braket{v_1,v_2,\ldots,v_N}{w_1,w_2,\ldots,w_N}=
\prod_{i=1}^N \braket{v_i}{w_i}.
$$
We denote
by $\wedge^N V$ the $N$th exterior power of $V$, i.e., the
subspace of $\cH$ consisting of the antisymmetric tensors.
We refer to vectors $\ket{\psi}\in\wedge^N V$ as
$N$-{\em vectors}.
We shall often identify the $N$-vectors with antisymmetric
tensors by using the embedding $\wedge^N V\to\cH$ given by:
 \bea
 \label{ea:isomorphism}
 \ket{v_1}\we\cdots\we\ket{v_N}
 \to
 \sum_{\sigma\in S_N}\mathrm{sgn}(\sigma)
 \ket{v_{\sigma(1)},v_{\sigma(2)},\ldots,v_{\sigma(N)}},
 \eea
where $S_N$ is the symmetric group on $N$ letters and
$\mathrm{sgn}(\sigma)$ the sign of the permutation $\sigma$.
We say that an $N$-vector $\ket{\psi}$ is {\em decomposable} if
it can be written as
$$
\ket{\psi}=|v_1\rangle\wedge \cdots  \wedge|v_N\rangle
$$
for some $\ket{v_i}\in V$. For brevity we may also write
$$
\ket{\psi}=|v_1\wedge \cdots  \wedge v_N\rangle
$$
when no confusion arises. Decomposable states are also called
Slater determinant states or Slater determinants in physics. For the
sake of brevity we shall mainly use the name of decomposable states.
Unless stated otherwise, the states will not be normalized for
convenience.

If $W$ is a vector subspace of $V$, then $\wedge^N W$ is a vector
subspace of $\wedge^N V$. Given a $\ket{\psi}\in\wedge^N V$, there
exists the smallest subspace $W\sue V$ such that
$\ket{\psi}\in\wedge^N W$. We shall refer to this subspace as the
{\em support} of $\ket{\psi}$. By Eq. \eqref{ea:isomorphism}, this
is equal to the range of the reduced density matrix of any single
particle of $\ket{\ps}$. In the case when $\ket{\psi}\ne0$ is
decomposable, say $\ket{\psi}=|v_1\wedge \cdots
\wedge v_N\rangle$, then its support is the subspace spanned by the
vectors $\ket{v_i}$, $i=1,\ldots,N$.

In general, an $N$-vector can always be written (not uniquely) as a
sum of decomposable $N$-vectors. With this notation, the inner
product in $\wedge^N V$ is characterized by the fact that the
inner product of two decomposable $N$-vectors
$\ket{\ph}=|w_1\wedge\cdots\wedge w_N\rangle$ and
$\ket{\psi}=|v_1\wedge\cdots\wedge v_N\rangle$ is equal
to the determinant of the $N\times N$ matrix $[\braket{w_i}{v_j}]$.
It is also equal to $\frac{1}{N!} \braket{\ph}{\psi}$ where the inner product is computed in the tensor space by using
the identification given by Eq. \eqref{ea:isomorphism}.
Hence, $\ket{\ph}\perp\ket{\psi}$ when this determinant is zero.

The 2-vectors $\ket{\psi}\in\we^2 V$ are often identified with
antisymmetric matrices of order $M$ via the isomorphism which
assigns to any decomposable 2-vector
$\ket{\psi}=\ket{v\we w}$ the matrix
$\ket{v}(\ket{w}^T)-\ket{w}(\ket{v}^T)$. Under this isomorphism,
the nonzero decomposable 2-vectors correspond to antisymmetric
matrices of rank 2.

We shall use the well known exterior multiplication map from
$\wedge^n V \times \wedge^N V$ into $\wedge^{N+n} V$, which is a
complex bilinear map. Note that this map is identically zero if
$N+n>M$. Let us also introduce the `partial inner product'
operations. Given an $n$-vector $\ket{\ph}$ and an $N$-vector
$\ket{\psi}$ with $n\le N$, we can define their partial inner
product $\braket{\ph}{\psi}\in\wedge^{N-n} V$. In other words, this
is a map
$$\wedge^n V \times \wedge^N V \to \wedge^{N-n} V$$
which is antilinear in the first argument $\ket{\ph}$ and complex
linear in the second argument $\ket{\psi}$. Thus, it suffices to
write the definition in the special case when $\ket{\ph}$ and
$\ket{\psi}$ are decomposable, say $\ket{\ph}=|w_1\wedge
\cdots  \wedge w_n\rangle$ and $\ket{\psi}=|v_1\wedge \cdots
\wedge v_N\rangle$. In that case we have
 \bea \label{eq:PartInnPr}
\braket{\ph}{\psi} &=& \sum_\s \mathrm{sgn}(\s)\left(
\prod_{i=1}^n
\braket{w_i}{v_{\s{i}}} \right) |v_{\s(n+1)}\wedge \cdots
\wedge v_{\s(N)}\rangle,
 \eea
where the summation is over all permutations $\s\in S_N$ such
that the sequence $\s(n+1),\ldots,\s(N)$ is increasing. This
definition is a modification of the one given in \cite[p. A III.166]{nb1970} for the partial bilinear inner product. We have omitted the overall factor $(-1)^{n(n-1)/2}$,
and replaced the bilinear inner product with the Hilbert space
inner product. More importantly, it appears that there is an
error (or misprint) in Bourbaki because the right hand side of their formula is not antisymmetric in the arguments $x_i^*$. We had to modify their formula by requiring that only the second one of their two sequences of $\s$-values is increasing.
The Eq. \eqref{eq:PartInnPr} can be rewritten as
 \bea \label{eq:PartInnPr-opet}
\braket{\ph}{\psi} &=& \sum_\s \mathrm{sgn}(\s)
\braket{ w_1\wedge \cdots \wedge w_n }
{ v_{\s{(1)}}\wedge \cdots \wedge v_{\s{(n)}} }
 \ket{ v_{\s(n+1)}\wedge \cdots \wedge v_{\s(N)} },
 \eea
where the summation now is over the permutations $\s\in S_N$ such
that both sequences $\s(1),\ldots,\s(n)$ and
$\s(n+1),\ldots,\s(N)$ are increasing.

If $n=1$ then the formula reads as follows:
 \bea \label{eq:PartInnPr-1}
\braket{w_1}{\psi} &=& \sum_{i=1}^N (-1)^{i-1} \braket{w_1}{v_i}
\ket{v_1 \wedge\cdots\wedge v_{i-1}\wedge v_{i+1}
\wedge\cdots\wedge v_N}.
 \eea
If we define the interior product $\iota(w)$ to be the
linear map $\wedge^N V\to\wedge^{N-1} V$ which sends $\ket{\psi}$
to $\braket{w}{\psi}$, then one can check that
Eq. \eqref{eq:PartInnPr} can be rewritten as follows:
 \bea
\braket{\ph}{\psi} &=& \iota(w_n)\circ\cdots\circ\iota(w_1)
(|v_1\wedge \cdots \wedge v_N\rangle).
 \eea

By using the embedding \eqref{ea:isomorphism}, the partial inner
product given by Eq. \eqref{eq:PartInnPr} is a positive scalar
multiple of the inner product $\braket{\ph}{\psi}$ where both $\ket{\ph}$ and $\ket{\ps}$ are viewed as antisymmetric tensors.  
Note that in the case $n=N$, the partial inner
product in Eq. \eqref{eq:PartInnPr} is the same as the inner
product on $\wedge^N V$ defined at the beginning of this section.

\subsection{Local unitary equivalence}

If $A_i:V\to V$, $i=1,\ldots,N$, are linear operators, then their
tensor product $\otimes_{i=1}^N A_i$ will be identified with the
unique linear operator on $\cH$ which maps
$$
 \ket{v_1}\otimes\ket{v_2}\otimes
\cdots\otimes\ket{v_N} \to
A_1\ket{v_1}\otimes A_2\ket{v_2}\otimes
\cdots\otimes A_N\ket{v_N},
\quad \ket{v_i}\in V,\quad(i=1,2,\ldots,N).
$$
If $A_1=A_2=\cdots=A_N=A$, then we shall write $\otimes^N A$
or $A^{\otimes N}$ instead of $\otimes_{i=1}^N A_i$ and refer to
$\otimes^N A$ as the $N$th tensor power of $A$.
If $\ket{\psi}\in\wedge^N V$ then we also have
$\otimes^N A(\ket{\psi})\in\wedge^N V$. Consequently, we
can restrict the operator $\otimes^N A$ to obtain a linear
operator on $\wedge^N V$, which we denote by $\wedge^N A$ or
$A^{\wedge N}$. We refer to $\wedge^N A$ as the $N$th exterior
power of $A$. Explicitly, we have
$$
\wedge^N A (\ket{v_1\wedge v_2\wedge
\cdots\wedge v_N})=
\ket{Av_1\wedge Av_2\wedge
\cdots\wedge Av_N}.
$$

The general linear group $G:=\GL(V)$ acts on $\cH$ by the
so called {\em diagonal action}:
$$
A\cdot(\ket{v_1}\otimes\ket{v_2}\otimes
\cdots\otimes\ket{v_N})=
A\ket{v_1}\otimes A\ket{v_2}\otimes
\cdots\otimes A\ket{v_N}, \quad A\in G,
\quad \ket{v_i}\in V,\quad(i=1,2,\ldots,N).
$$
In other words, $A\in G$ acts on $\cH$ as $\otimes^N A$. Similarly, we have the action of $G$ on $\wedge^N V$ where
$A\in G$ acts as $\wedge^N A$. For convenience we shall
abbreviate both of these two actions by a `$\cdot$', i.e.,
we have
 \bea
A\cdot\ket{v_1,\ldots,v_N} &=& \ket{Av_1,\ldots,Av_N}, \\
A\cdot\ket{v_1\wedge\cdots\wedge v_N} &=&
\ket{Av_1\wedge\cdots\wedge Av_N}.
 \eea
Both of these actions can be restricted to the unitary group
$\Un(V)$ of $V$.

We shall say that two $N$-vectors $\ket{\phi}$ and $\ket{\psi}$
are {\em equivalent} if they belong to the same $G$-orbit, i.e.,
$\ket{\psi}=A\cdot\ket{\phi}$ for some $A\in G$. We shall also
say that they are {\em unitarily equivalent} or
{\em LU-equivalent} if such $A$ can be chosen to be unitary.
For example, any decomposable $N$-vector is unitarily equivalent to a
scalar multiple of $\ket{1}\we\cdots\we\ket{N}$.

In practice, LU can be realized by a Hamiltonian
\begin{equation}
\label{eq:H} H=\sum_{j=1}^N H_j,
\end{equation}
where $H_j$ acts on the $j$th particle,
$H_j\ket{v_1,\ldots,v_N}=\ket{v_1,\ldots,H_jv_j,\ldots,v_N}$. The evolution of the system is unitary, i.e.
\begin{equation}
e^{-iHt}=\exp{(\sum_{j=1}^N -iH_jt)}=\bigotimes_{j=1}^N U_j,
\end{equation}
with $U_j=e^{-iH_jt}$. The Hamiltonian of the form Eq.~\eqref{eq:H} is usually called a single-particle Hamiltonian.
For a fermionic system, we have $H_1=H_2=\cdots=H_N=H$, hence
$U_1=U_2=\cdots=U_N=U$. Hence the evolution of the system is indeed the fermionic LU $\wedge^N U$.

In the case $N=2$, we have the canonical form for unitary
equivalence.
 \bl
 \label{le:2ferM}
If $N=2$, then any 2-vector $\ket{\psi}$ is unitarily equivalent
to $\sum^k_{i=1}c_i \ket{2i-1}\we\ket{2i}$ for some
$c_1\ge\cdots\ge c_k>0$ and some integer $k\ge0$. Moreover, the
coefficients $c_i$ are uniquely determined by $\ket{\psi}$.
 \el
 \bpf
The first assertion follows easily from the antisymmetric version
of Takagi's theorem, see \cite[p. 217]{hj85}.
To prove the second assertion, we denote by $K$ the antisymmetric
matrix corresponding to $\ket{\psi}$. If $U$ is a unitary
operator on $V$, then the 2-vector $U\ket{\psi}$ is represented
by the antisymmetric matrix $UKU^T$. It follows that the
Hermitian matrix $KK^*$ is transformed to $UKK^*U^\dag$, and so
its eigenvalues are independent of $U$. Since these eigenvalues
are $-c_1^2,-c_1^2,\ldots,-c_k^2,-c_k^2,0,\ldots,0$, the second
assertion follows.
 \epf

It can be seen that the Slater decomposition for 2-vectors is
essentially the same as the well-known Schmidt decomposition for
bipartite pure states. While the later can be easily generalized to
multipartite systems by regarding them as bipartite systems, the
same method does not work for the former, because the two parties of
the bipartite system would have different dimensions and this system
cannot be antisymmetric.

We remark that this lemma provides a complete classification of
$\Un(V)$-orbits on $\wedge^2 V$. Let us also mention that the
antisymmetric version of Takagi's theorem used in the above proof
was reproved recently in \cite[Lemma 1]{SCK+01}. However the authors
missed to observe that this result can be used to construct the
canonical form for bipartite fermionic states. Similarly, the
symmetric version of Takagi's theorem provides a canonical form for
pure states of two bosons, which is rediscovered recently
in~\cite{LZL+01,PY01}. See also \cite{hks12} where more advanced
mathematical tools were used to solve this classification problem as
well as its analogue for two bosons.

In the case when $V$ is 5-dimensional, there is also a simple
canonical form for 3-vectors under unitary equivalence. \bcr
 \label{cr:3fer5}
If $\dim V=5$ then every 3-vector $\ket{\psi}$ is unitarily
equivalent to
$(c_1\ket{1}\we\ket{2}+c_2\ket{3}\we\ket{4})\we\ket{5}$ for some
$c_1\ge c_2\geq 0$. Moreover, the coefficients $c_1$ and $c_2$ are
uniquely determined by $\ket{\psi}$. \ecr \bpf It is well known that
we can write $\ket{\psi}=\ket{\ph}\we\ket{v}$ for some 2-vector
$\ket{\ph}$ and some $\ket{v}\in V$. Clearly we may assume that
$\|v\|=1$ and that the support of $\ket{\ph}$ is orthogonal to
$\ket{v}$. Hence, by applying Lemma \ref{le:2ferM}, we may assume
that $\ket{v}=\ket{5}$ and
$\ket{\ph}=c_1\ket{1}\we\ket{2}+c_2\ket{3}\we\ket{4}$, $c_1\ge
c_2\ge0$. It remains to prove the uniqueness assertion. Note that
$\ket{\psi}$ is decomposable if and only if $c_2=0$, in which case
we have $c_1=\|\psi\|$. So, we may assume that $c_2>0$. Assume that
we also have $U\cdot\ket{\psi}=\ket{\psi'}=\ket{\ph'}\we\ket{5}$,
where $\ket{\ph'}=c'_1\ket{1}\we\ket{2}+c'_2\ket{3}\we\ket{4}$ and
$c'_1\ge c'_2>0$. Since $\ket{5}$ spans the kernel of the linear map
$V\to\we^4 V$ sending $\ket{x}\to\ket{\psi}\we\ket{x}$ and the same
is true for $\ket{\psi'}$, we deduce that $U\ket{5}=z\ket{5}$ where
$z$ is some phase factor (see \cite[Proposition 11.28]{bh07}). We
deduce that $\ket{\ph}$ and $\ket{\ph'}$ are unitarily equivalent,
and the uniqueness assertion follows from Lemma \ref{le:2ferM}. \epf

\subsection{Universal spaces and single occupancy states}
\label{se:UnivSO}

Let $G$ be a connected compact Lie group and $\cV$ a real
$G$-module. In other words, $\cV$ is a finite-dimensional real
vector space and we have a linear representation $\r$ of $G$
on $\cV$. For a vector subspace $\cW\subseteq\cU$, we say that
$\cW$ is {\em universal} if every $G$-orbit in $\cV$ meets
$\cW$ in at least one point. As an example consider the
conjugation action of the unitary group $\Un(d)$ on the space
of $d\times d$ complex matrices: $(A,X)\to AXA^{-1}$ where
$A\in\Un(d)$. Then the well known Schur's triangularization
theorem can be simply expressed by saying that the subspace
of upper triangular matrices is universal for this action.
There are many other interesting examples of universal
subspaces that occur in mathematics and physics. The first
question one may ask about the universal subspaces is: what
is the minimum dimension of a universal subspace? We can give
a simple lower bound for this dimension. Denote by $G_\cW$
the stabilizer of $\cW$ in $G$, i.e., the subgroup of $G$
consisting of all elements $g\in G$ such that $\r(g)$ maps
$\cW$ onto itself. Then we have the following lower bound
(see \cite[Lemma 4.1]{ad10}). If $\cW$ is a universal subspace
then
 \bea
\dim\cW \ge \dim\cV-\dim G/G_\cW,
 \eea
where the dimension of the quotient space $G/G_\cW$ is equal
to $\dim G -\dim G_\cW$. (All the dimensions here are to be
interpreted as the dimensions of real manifolds.) A case of
special interest is when $G_\cW$ contains a maximal torus,
say $T$, of $G$. Then $G/T$ has even dimension, say $2m$,
and we obtain the inequality $\dim\cW\ge\dim\cV-2m$. Assume
now that in fact $\cV$ is a complex vector space of dimension
$d$, and so of real dimension $2d$, while $\cW$ is still
assumed to be just a real vector subspace of $\cV$. If $\cW$
is universal, we must have $\dim\cW\ge2(d-m)$.

In our application we shall take $\cV=\wedge^N V$, where $V$
is a complex Hilbert space of dimension $M$, and so
$d=\binom{M}{N}$. We also specify that $G=\Un(M)$ and so
we have $m=\binom{M}{2}$. In particular,
for $N=3$, the above inequality becomes
$\dim_\bR\cW\ge M(M-1)(M-5)/3$. If $\cW$ is also a complex
subspace, we obtain that $\dim_\bC\cW\ge M(M-1)(M-5)/6$.

In order to define the important notion of single occupancy states
in a fermionic system, we partition the set of positive integers
into pairs $\{2i-1,2i\}$, $i=1,2,\ldots$. We shall refer to these
pairs as {\em standard pairs}. For any positive integer $i$ we
denote by $\overline{i}$ the standard pair to which $i$ belongs.
Thus, we have $\overline{i}=\{i,i+1\}$ when $i$ is odd and
$\overline{i}=\{i-1,i\}$ when $i$ is even. We shall say that a
finite sequence of positive integers $(i_1,i_2,\ldots,i_k)$ is a
{\em single occupancy sequence} if the standard pairs
$\overline{i_1},\overline{i_2},\ldots,\overline{i_k}$ are distinct.

Next we fix an o.n. basis $\{\ket{i}:i=1,\ldots,M\}$ of the Hilbert
space $V$. As explained below Fig \ref{fig:4sites}, each of our standard pairs
$\overline{i}$ associated with basis vectors $\{\ket{i},\ket{i+1}\}$ for $i$ odd,
then correspond to $\{\ket{g_{j}}\otimes\ket{\uparrow},\ket{g_j}\otimes\ket{\downarrow}\}$,
where $j=(i+1)/2$. Similar for $i$ even, the standard pair
$\overline{i}$ associated with basis vectors $\{\ket{i-1},\ket{i}\}$,
then correspond to $\{\ket{g_{j}}\otimes\ket{\uparrow},\ket{g_j}\otimes\ket{\downarrow}\}$,
where $j=i/2$.

The ${M\choose N}$ decomposable vectors $\ket{i_1\we\cdots\we i_N}$, $1\le
i_1<\cdots <i_N\le M$, form an o.n. basis of $\we^N V$. Hence, the
dimension of the space $\we^N V$ is $\binom{M}{N}$. If the sequence
$(i_1,i_2,\ldots,i_N)$ is a single occupancy sequence then we shall
say that $\ket{i_1\we\cdots\we i_N}$ is a {\em basic single
occupancy $N$-vector} (BSOV). Note that if $M$ is odd then
$\overline{M}=\{M,M+1\}$ and so there is no state in $V$ that is
paired with $\ket{M}$. Finally, we say that an $N$-vector
$\ket{\psi}\in\wedge^N V$ is a {\em single occupancy vector} (SOV),
relative to the above o.n. basis, if it is a linear combination of
the BSOVs. Thus, the set of all SOVs is a vector subspace (the SOV
subspace) of $\wedge^N V$ and the set of $BSOV$s is an o.n. basis of
this subspace. Clearly, this definition depends on the choice of the
o.n. basis of $V$. A simple counting shows that the number of BSOVs
is $2^N\binom{K}{N}$ if $M=2K$ is even and
$2^N\binom{K}{N}+2^{N-1}\binom{K}{N-1}$ if $M=2K+1$ is odd. For the
case $N=K$ and $M=2K$, if one ensures single occupancy, then there
are indeed only two possible states per site (spin up and spin
down). In this case, any configuration has a natural correspondence
to an $N$-qubit state.

\begin{lemma}
 \label{le:N=2SOV}
For $N=2$, $M>3$, any state is unitarily equivalent to an SOV.
 \end{lemma}
\bpf
This is a direct corollary of Lemma~\ref{le:2ferM}.
\epf

\section{Universal subspaces for $N=3$}
\label{sec:USN3}

\subsection{The SOV-subspace is universal for $N=3$}

The main result of this section is that any tripartite antisymmetric state of dimension $M\ge5$ is
unitarily equivalent to a single occupancy state.
Recall we say that a vector subspace
$\mathcal{W}\subseteq\we^3 V$ is {\em universal} if every
$\ket{\psi}\in\we^3 V$ is unitarily equivalent to some
$\ket{\ph}\in \mathcal{W}$. We denote the unitary group of $V$ by
$\Un(V)$. If an o.n. basis of $V$ is fixed, we shall identify
$\Un(V)$ with the group $\Un(M)$ of unitary matrices of order
$M$.

\bt \label{thm:BSOV}
Let $V$ be a complex Hilbert space of dimension $M\ge5$,
$\{\ket{i}:i=1,2,\ldots,M\}$ an o.n. basis of $V$, and
let $e_{ijk}=\ket{i}\we\ket{j}\we\ket{k}$. Then the
SOV-subspace is universal, i.e., any 3-vector
$\ket{\psi}\in\we^3 V$ is unitarily equivalent to an SOV.
 \et

 \bpf
The case $M=5$ follows from Corollary \ref{cr:3fer5}. So, we
shall assume that $M\ge6$. We shall write $M=2K$ when $M$ is even
and $M=2K+1$ when $M$ is odd.

We apply \cite[Theorem 4.2]{ad10} to the problem at hand. Let us
first explain what this theorem asserts in this concrete case. We
consider the representation $\rho$ of $\Un(M)$ on the space
$\cV:=\wedge^3 V$. We denote by $T$ the maximal torus of $\Un(M)$
consisting of the diagonal matrices. Thus, if $x\in T$ then
$x=\diag(\xi_1,\ldots,\xi_M)$ where each $\xi_i$ is a complex
number of unit modulus. Each basis vector $e_{ijk}$ is an
eigenvector of $T$. Indeed, we have
$x\cdot e_{ijk}:=\rho(x)(e_{ijk})=\xi_i\xi_j\xi_k e_{ijk}$.

A character, $\chi$, of $T$ is a continuous homomorphism into
the circle group $S^1:=\{z\in\bC:~|z|=1\}$. They form an
abelian group under multiplication, but it is convenient to
use the additive notation. This means that if $\chi'$
and $\chi''$ are characters of $T$, then their sum
$\chi=\chi'+\chi''$ is defined by $\chi(x)=\chi'(x)\chi''(x)$.
The map $\chi_i:T\to S^1$, defined by $\chi_i(x)=\xi_i$, is
obviously a character of $T$. The group of all characters
of $T$ is the free abelian group with basis
$\{\chi_i:1\le i\le M\}$. By using this notation, we have
$x\cdot e_{ijk}=(\chi_i+\chi_j+\chi_k)(x) e_{ijk}$ and we
say that the character $\chi_{ijk}:=\chi_i+\chi_j+\chi_k$ is
the weight of the eigenvector $e_{ijk}$. An eigenvector which
is fixed by $T$ would have weight 0, but in $\cV$ there are
no such eigenvectors.

We embed the character group of $T$
into the polynomial ring $\cP:=\bR[x_1,\ldots,x_M]$ by sending
$\chi_i\to x_i$ for each $i$. Let $\cI$ be the ideal of
$\bR[x_1,\ldots,x_M]$ generated by the elementary symmetric
functions $\sigma_1,\ldots,\sigma_M$ of the $x_i$.
Let $\cB$ be the o.n. basis of $\cV$ consisting of all
$e_{ijk}$. Denote by $\cB_s$ the subset of $\cB$ consisting
of all BSOVs. Let $\cU$ be a subspace of $\cV$ spanned by a
subset $\cB'\subseteq\cB$. We define its characteristic
polynomial $f_\cU$ to be the product of all linear
polynomials $x_i+x_j+x_k$ taken over all triples $(i,j,k)$
such that $e_{ijk}\notin\cB'$. Then \cite[Theorem 4.2]{ad10}
asserts that the subspace $\cU$ is universal if
$f_\cU\notin\cI$.

To prove (a) we shall take $\cB'=\cB_s$.
Consequently, we have $\cU$ equals the SOV-subspace and the polynomial
$q_0:=f_\cU$ is given explicitly by the formula
\bea \label{eq:Pol-q0}
 q_0 &=& \prod_{i=1}^K \prod_{j\ne 2i-1,2i}
(x_{2i-1} + x_{2i} + x_j).
\eea

We shall need another important fact concerning the ideal $\cI$.
First observe that $\cI$ is a homogeneous ideal, i.e., it is the
sum of the intersections $\cI_d:=\cI\cap\cP_d$, where $\cP_d$
is the subspace of $\cP$ consisting of all homogeneous polynomials of degree $d$. We introduce an inner product
$\braket{\cdot}{\cdot}$ on $\cP$ by declaring that the basis of $\cP$ consisting of all monomials is an o.n. basis.
In particular, $\cP_d\perp\cP_e$ if $d\ne e$.
We denote by $S_M$ the symmetric group on $M$ letters which
permutes the $M$ variables $x_i$ and point out that it
preserves the above inner product on $\cP$.
The following polynomial of degree $\d:=M(M-1)/2$ (with the well known expansion)
 \bea \label{eq:Pol-p}
 p:=\prod_{1\le i<j\le M}(x_j-x_i)=\sum_{\s\in S_M}\text{sgn} (\s) x_{\s 1}^0 x_{\s 2}^1 x_{\s 3}^2 \cdots x_{\s M}^{M-1},
 \eea
will play an important role. Namely, for any $f\in\cP_\d$ we
know that $f\in\cI$ if and only if $\braket{f}{p}=0$, see \cite[Theorem 4.2]{ad10}. For $f,g\in\cP$, we shall write $f\equiv g$ if $g-f\in\cI$ and in that case we say that $f$ and $g$ are congruent modulo $\cI$. Thus, if $f\equiv g$ then $\braket{f}{p}=\braket{g}{p}$. Note that the degree of $q_0$ is strictly less than $\d$. For that reason we shall introduce the polynomial
$q=\mu q_0$ of degree $\d$, where $\mu=x_1 x_3\cdots x_{M-3} x_{M-1}$ if $M$ is even and $\mu=(x_1 x_3\cdots x_{M-2})^2$ if $M$ is odd. To prove (b) we shall prove that $q\notin\cI$ (which
implies that $q_0\notin\cI$).

From the formula \eqref{eq:Pol-q0}, we obtain that
\bea \label{eq:Pol-q0-1}
 q_0 &=& \prod_{i=1}^K \sum_{j=0}^{M-2} (x_{2i-1}+x_{2i})^{M-2-j}
\s^{(i)}_j,
\eea
where $\s^{(i)}_j$ is the $j$th elementary symmetric function
of the variables $x_k$ with $k\ne 2i-1,2i$.
(By convention, $\s^{(i)}_0=1$ and $\s^{(i)}_{-1}=0$.)
By using the obvious recurrence formula
$\s_j=\s^{(i)}_j+(x_{2i-1}+x_{2i})\s^{(i)}_{j-1}+x_{2i-1}x_{2i}
\s^{(i)}_{j-2}\equiv 0$, we obtain that
\bea \label{eq:sigma}
\s^{(i)}_j\equiv (-1)^j\sum_{k=0}^j x_{2i-1}^{j-k}x_{2i}^k,
\quad j=0,1,\ldots,M-2.
\eea
It follows that
\bea \label{eq:Pol-q0-2}
 q_0 &\equiv& \prod_{i=1}^K \left( \sum_{j=0}^{M-2} (-1)^j
(x_{2i-1}+x_{2i})^{M-2-j} \sum_{k=0}^j x_{2i-1}^{j-k}x_{2i}^k
\right).
\eea

We shall need the expansion
\bea \label{eq:Expansion}
\sum_{j=0}^{M-2}(-1)^j (x+y)^{M-2-j}\sum_{k=0}^j x^{j-k}y^k =
\sum_{k=0}^{M-2} a_k^{(M)} x^{M-2-k} y^k,
\eea
where $x$ and $y$ are commuting independent variables. Note that $a_k^{(M)}=a_{M-2-k}^{(M)}$ for each $k$, and
$a_0^{(M)}=a_{M-2}^{(M)}=1$ for $M$ even while
$a_0^{(M)}=a_{M-2}^{(M)}=0$ for $M$ odd. For convenience, we
set $a_{M-1}^{(M)}=0$.
We now distinguish two cases.

Case 1: $M$ is even. By multiplying Eq. \eqref{eq:Pol-q0-1} by
$\mu$ and by using Eqs. \eqref{eq:sigma} and
\eqref{eq:Expansion}, we obtain that
 \bea \label{eq:Pol-q0-4}
 q\equiv\prod_{i=1}^K
\sum_{k=0}^{M-2}a_k^{(M)} x_{2i-1}^{M-1-k}x_{2i}^k.
 \eea
By setting $s_{i,k}=a_k^{(M)} x_{2i-1}^{M-1-k}x_{2i}^k
+a_{M-1-k}^{(M)} x_{2i-1}^k x_{2i}^{M-1-k}$
for $k=0,1,\ldots,K-1$, we obtain that
 \bea
q \equiv \prod_{i=1}^K(s_{i,0}+s_{i,1}+\cdots+s_{i,K-1})
= \sum_f \prod_{i=1}^K s_{i,f(i-1)},
\eea
where $f$ runs through all functions mapping the set
$\{0,1,\ldots,K-1\}$ into itself. Let us denote by $q_f$
the summand corresponding to $f$.
If $f$ is not a permutation, then each monomial that occurs in
the expansion of $q_f$ will contain two variables with the
same exponent, and so the inner product $\braket{q_f}{p}$
vanishes. If two of such functions, say $f$ and $g$ are
permutations then there is an even permutation of the variables
which sends $q_f$ to $q_g$. We deduce that
$\braket{q_f}{p}=\braket{q_g}{p}$, and so
$\braket{q}{p}=K!\braket{q_{\rm id}}{p}$
where id is the identity permutation.

If the variables $x_{2i-1}$ and $x_{2i}$ do not occur in a
polynomial $r\in\cP$, then by using the transposition which
interchanges $x_{2i-1}$ and $x_{2i}$, we obtain that
$\braket{x_{2i-1}^{M-1-k}x_{2i}^k r}{p}=
-\braket{x_{2i-1}^k x_{2i}^{M-1-k}r}{p}$, and so
 \bea
\braket{s_{i,k}r}{p}=(a_{M-1-k}^{(M)}-a_k^{(M)})
\braket{x_{2i-1}^k x_{2i}^{M-1-k}r}{p}.
 \eea
Since $q_{\rm id}=s_{10}s_{21}s_{32}\cdots s_{K,K-1}$, we obtain
that
\bea \label{eq:SkalPr-3}
\braket{q}{p}=K!\prod_{i=0}^{K-1} (a_{M-1-i}^{(M)}-a_i^{(M)}).
\eea
By Lemma \ref{le:Koef-a} below, we have $\braket{q}{p}\ne0$.

Case 2: $M$ is odd. The proof follows closely the one for
Case 1 and we shall only briefly sketch the main steps. Recall
that in this case $a_0^{(M)}=a_{M-2}^{(M)}=0$. By multiplying
Eq. \eqref{eq:Pol-q0-1} by the new $\mu$ and by using Eqs.
\eqref{eq:sigma} and \eqref{eq:Expansion}, we obtain that
 \bea \label{eq:Poli-q0-4}
 q \equiv \prod_{i=1}^K \left( x_{2i-1}^2
\sum_{k=1}^{M-2} a_k^{(M)} x_{2i-1}^{M-2-k} x_{2i}^k \right)=
\prod_{i=1}^K\sum_{k=1}^{M-2}a_k^{(M)} x_{2i-1}^{M-k} x_{2i}^k.
 \eea
By setting $s'_{i,k}=a_k^{(M)} x_{2i-1}^{M-k}x_{2i}^k
+a_{M-k}^{(M)} x_{2i-1}^k x_{2i}^{M-k}$ for $k=1,2,\ldots,K$,
we have
 \bea
q \equiv \prod_{i=1}^K(s'_{i,1}+s'_{i,2}+\cdots+
s'_{i,K}) = \sum_f \prod_{i=1}^K s'_{i,f(i)},
 \eea
where $f$ runs through all functions mapping the set
$\{1,2,\ldots,K\}$ into itself. Let us denote by $q_f$ the
summand corresponding to $f$. If $f$ is not a permutation, then
$\braket{q_f}{p}=0$. When $f$ is a permutation, this inner
product is independent of $f$. Thus,
$\braket{q}{p}=K!\braket{q_{\rm id}}{p}$ where
$q_{\rm id}=s'_{11}s'_{22}s'_{33}\cdots s'_{K,K}$. By using a similar argument as in Case 1, we obtain that
\bea \label{eq:SkalPr-4}
\braket{q}{p}=K!\prod_{i=1}^K (a_{M-i}^{(M)}-a_i^{(M)}).
\eea
By Lemma \ref{le:Koef-a}, $\braket{q}{p}\ne0$ which completes the proof.
 \epf

\begin{lemma}  \label{le:Koef-a}
For any $M>4$, let $a_i^{(M)}$ be the coefficients defined by
Eq. \eqref{eq:Expansion}.
Then
\begin{enumerate}
\item[1.] if $M$ is even, then $a_p^{(M)}\neq a_{M-1-p}^{(M)}$ for any $p$,
\item[2.] if $M$ is odd, then $a_p^{(M)}\neq a_{M-p}^{(M)}$ for any $p$.
\end{enumerate}
 \end{lemma}

\begin{proof}
Since the polynomial defining the coefficients $a_i^{(M)}$ is
symmetric, we have $a_p^{(M)}=a_{M-2-p}^{(M)}$. By expanding
this polynomial we obtain that
\begin{eqnarray}
&& \sum\limits^{M-2}_{j=0} (-1)^j (x+y)^{M-2-j} (\frac{x^{j+1}-y^{j+1}}{x-y})\\
&=& \sum\limits^{M-2}_{j=0} (-1)^j \sum\limits_{k=0}^{M-2-j} {M-2-j\choose k} x^{M-2-j-k}y^k \sum\limits_{t=0}^j x^{j-t} y^t\\
&=&\sum\limits^{M-2}_{j=0}\sum\limits_{k=0}^{M-2-j} \sum\limits_{t=0}^j(-1)^j  {M-2-j\choose k} x^{M-2-t-k}y^{k+t},
\end{eqnarray}
and so

\begin{eqnarray*}
a_p^{(M)}&=&\sum\limits_{k=0}^p \sum\limits_{j=0}^{M-2-p}(-1)^{M-2-k-j}{k+j\choose k}\\
&=&\sum\limits_{k=0}^p \sum\limits_{j=0}^{M-2-p}(-1)^{M-2-k-j}({k+j-1\choose k-1}+{k+j-1\choose k})+(-1)^{M-2}\\
&=&\sum\limits_{k=0}^{p-1} \sum\limits_{j=0}^{M-2-p}(-1)^{M-2-(k+1)-j}{k+j\choose k}+\sum\limits_{k=0}^p \sum\limits_{j=0}^{M-3-p}(-1)^{M-2-k-(j+1)}{k+j\choose k}+(-1)^{M-2}\\
&=&a_{p-1}^{(M-1)}+a_p^{(M-1)}+(-1)^{M-2}.
\end{eqnarray*}

Then the values of $a_p^{(M)}$ can be easily computed, see
Table \ref{table:nonlin}.
\begin{table}[ht]
\caption{Values of $a_p^{(M)}$} 
\centering  
\begin{tabular}{c c c c c c c c } 
\hline\hline                        
$M$ & $a_0^{(M)}$ & $a_1^{(M)}$ & $a_2^{(M)}$ & $a_3^{(M)}$ & $a_4^{(M)}$ & $a_5^{(M)}$ & $a_6^{(M)}$\\ [0.5ex] 
\hline                  
4 & 1 & 1 & 1 & 0 &  &  &  \\ 
5 & 0 & 1 & 1 & 0 &  &  &  \\
6 & 1 & 2 & 3 & 2 & 1 & 0 &  \\
7 & 0 & 2 & 4 & 4 & 2 & 0 & \\
8 & 1 & 3 & 7 & 9 & 7 & 3 & 1 \\ [1ex]      
\hline 
\end{tabular}
\label{table:nonlin} 
\end{table}

One may further observe that, for any even number $M$, $a_p^{(M)}\neq a_{M-1-p}^{(M)} \iff a_p^{(M)}\neq a_{p-1}^{(M)} \iff a_{p-1}^{(M-1)}+a_p^{(M-1)}\neq a_{p-2}^{(M-1)}+a_{p-1}^{(M-1)} \iff a_p^{(M-1)}\neq a_{p-2}^{(M-1)} \iff a_p^{(M-1)}\neq a_{M-1-p}^{(M-1)}$.

Therefore, the two parts in our lemma are indeed equivalent, hence we only need to focus on the case that $M$ is even.

For even $M=2K>4$ and $1\leq p\leq M-3$, we always have $a_0^{(M)}=1$ and
\begin{eqnarray}
a_p^{(M)}&=&a_{p-1}^{(M-1)}+a_p^{(M-1)}+1\\
&=&a_{p-2}^{(M-2)}+a_{p-1}^{(M-2)}-1+a_{p-1}^{(M-2)}+a_{p}^{(M-2)}-1+1\\
&=&a_{p-2}^{(M-2)}+2a_{p-1}^{(M-2)}+a_p^{(M-2)}-1.
\end{eqnarray}

We then show that $\{a_0^{(2K)}, a_1^{(2K)},\cdots, a_{K-1}^{(2K)}\}$ is an increasing sequence. It is obviously true for $K=3,4$. Assume our claim is true for $K=K_0$, let's look into the case $K=K_0+1$. We have $a_p^{(2K_0+2)}=a_{p-2}^{(2K_0)}+2a_{p-1}^{(2K_0)}+a_p^{(2K_0)}-1$ for any $1\leq p \leq K_0$. Then $\{a_0^{(2K_0+2)}, a_1^{(2K_0+2)},\cdots, a_{K_0-1}^{(2K_0+2)}\}$ is also an increasing sequence. Since $a_{K_0}^{(2K_0+2)}-a_{K_0-1}^{(2K_0+2)}=(a_{K_0-2}^{(2K_0)}+2a_{K_0-1}^{(2K_0)}+a_{K_0}^{(2K_0)}-1)-(a_{K_0-3}^{(2K_0)}+2a_{K_0-2}^{(2K_0)}+a_{K_0-1}^{(2K_0)}-1)=a_{K_0}^{(2K_0)}+a_{K_0-1}^{(2K_0)}-a_{K_0-2}^{(2K_0)}-a_{K_0-3}^{(2K_0)}=a_{K_0-1}^{(2K_0)}-a_{K_0-3}^{(2K_0)}>0$, the whole sequence $\{a_0^{(2K_0+2)}, a_1^{(2K_0+2)},\cdots, a_{K_0}^{(2K_0+2)}\}$ is also an increasing sequence. This completes our proof.
\end{proof}

\subsection{Universal subspaces of minimal dimensions}

It follows from Proposition \ref{prp:UnivCond} that any subspace
of $\wedge^3 V$ which is spanned by less than $M(M-1)(M-5)/6$
basic trivectors $e_{ijk}$ is not universal. It is natural to
raise the question whether there exist universal subspaces
spanned by exactly $M(M-1)(M-5)/6$ BSOVs. In the next proposition we show that this is indeed true when $M$ is even, and verify it for odd $M=7,9,11$.

\bpp \label{pp:BSOV} Let $V$ be a complex Hilbert space of
dimension $M$ and let $e_{ijk}=\ket{i}\we\ket{j}\we\ket{k}$
for $1\le i<j<k\le M$.

(a) If $M=2K\ge6$ is even, then the complex subspace spanned
by all BSOVs $e_{ijk}$ except the $K$ of them with indexes
 \bea
 \label{ea:evenM}
(1,3,6),(1,4,6)\quad{\textit and}\quad(1,2i-3,2i-1)\quad
3\le i\le K,
 \eea
is a universal subspace of minimum dimension, $M(M-1)(M-5)/6$.

(b) If $M=7,9,11$, then the complex subspace spanned by all BSOVs $e_{ijk}$ except the $M-1$ of them with indexes
 \bea \label{eq:neparnoM}
(1,4,M),(2,5,M)\quad{\textit and}\quad(i,i+2,M)
\quad1\le i\le M-3;
 \eea
is a universal subspace of minimum dimension, $M(M-1)(M-5)/6$.
 \epp
 \bpf (a) We can write an arbitrary $\ket{\psi}\in\we^3 V$ as
$\ket{\ps}=\sum_{i,j,k} c_{ijk} e_{ijk}$ where the summation
is over all triples $(i,j,k)$, $1\le i<j<k\le M$. First, by
Theorem \ref{thm:BSOV} we may assume that the coefficients
$c_{ijk}=0$ whenever $e_{ijk}$ is not a BSOV. For $i\le K$,
denote by $\Un_i$ the subgroup of $\Un(M)$ which fixes all basis
vectors $\ket{j}$ with $j\ne 2i-1,2i$. All subsequent LU
transformations will be performed by using these subgroups,
and so the above mentioned property of the coefficients
$c_{ijk}$ will be preserved.

Second, we shall prove by induction that, by using only the
$\Un_i$ operations, we can achieve our goal to make $c_{ijk}=0$
also when $(i,j,k)$ is one of the triples $(1,3,6)$, $(1,4,6)$ or $(1,2i-3,2i-1)$ with $3\le i\le K$.

If $K=3$, then $M=6$. Note that the Hilbert
space of 3-qubits can be isometrically embedded in this
fermionic system as the subspace spanned by the BSOVs:
 $$
\ket{i}\ox\ket{j}\ox\ket{k}\to\ket{i+1}\we\ket{j+3}\we\ket{k+5}, \quad i,j,k\in\{0,1\}.
 $$
The desired result then follows from the well known fact \cite{ajt01} that by performing LU transformations on any given pure state of three qubits one can make vanish any three coefficients (in the standard o.n. basis).
This is because after the embedding, the three-qubit LU can be viewed as a special case of the fermionic LU, where we allow only the fermionic LU transformations given by block-diagonal unitary in $\Un(6)$ with three $2\times 2$ blocks in $\Un(2)$.

Now let $K>3$. By the induction hypothesis, we may assume that the coefficients $c_{ijk}=0$ when $(i,j,k)$ is one of the triples $(1,3,6)$, $(1,4,6)$ or $(1,2i-3,2i-1)$ with $3\le i<K$.
We can choose $X\in\Un_K$ such that
$X(c_{1,M-3,M-1}\ket{M-1}+c_{1,M-3,M}\ket{M})\propto\ket{M}$.
This means that the coefficient of $e_{1,M-3,M-1}$ in
$X\cdot\ket{\ps}$ is 0. This completes the inductive proof.

(b) This proof is computational. All computations were performed by using Singular \cite{GPS}. We just have to apply
\cite[Theorem 4.2]{ad10} as explained in the proof of Theorem
\ref{thm:BSOV}. The cases $M=7$ and $M=9$ were straightforward,
and we omit the details. The values $\braket{q}{p}$ that we computed for $M=7,9,11$ are 48, 10368 and 12431232, respectively.
In the case $M=11$ we had first to eliminate the variable
$x_{11}$ in order to avoid the problem of running out of memory. We shall describe this elimination procedure in general.

Thus, let $M\ge7$ be an odd integer, and set $M=2K+1$.
We define the polynomial ring $\cP$, its ideal $\cI$, and the inner product $\langle\cdot,\cdot\rangle$ of polynomials as in
the proof of Theorem \ref{thm:BSOV}. The polynomials $q_0$ and
$p$ will be the same as in that proof, see Eqs. \eqref{eq:Pol-q0} and \eqref{eq:Pol-p}. The congruence of polynomials will be again modulo the ideal $\cI$. We shall also use the coefficients
$a^{(M)}_k$, $k=0,1,\ldots,M-2$, defined by the polynomial identity Eq. \eqref{eq:Expansion}. (These coefficients are the same as the ones used in Lemma \ref{le:Koef-a}.) Other symbols that we are going to introduce, like $\mu,q$ and $s_{i,k}$ will have different meaning from the same symbols used in the proof of Theorem \ref{thm:BSOV}. Since $a^{(M)}_0=a^{(M)}_{M-2}=0$, from Eqs. \eqref{eq:Pol-q0-2} and \eqref{eq:Expansion} we obtain that
 \bea
 q_0 &\equiv& \prod_{i=1}^K
\sum_{k=1}^{M-3}a_k^{(M)} x_{2i-1}^{M-2-k}x_{2i}^k \notag \\
&=& \prod_{i=1}^K \sum_{k=1}^{K-1} s_{i,k}=
\sum_f \prod_{i=1}^K s_{i,f(i)},
 \eea
where $f$ runs through all functions
$\{1,\ldots,K\}\to\{1,\ldots,K-1\}$ and
 \bea
s_{i,k} &=& (x_{2i-1}x_{2i})^k \left( a_k^{(M)} x_{2i-1}^{M-2-2k}
+a_{M-2-k}^{(M)} x_{2i}^{M-2-2k} \right) \notag \\
&=& x_{2i-1}x_{2i} s'_{i,k},
\quad k=1,\ldots,K-1.
\eea
Thus $q_0\equiv\sum_f q_f$ where
$q_f=x_1 x_2\cdots x_{M-1} q'_f$ and
$q'_f=\prod_{i=1}^K s'_{i,f(i)}$. The characteristic polynomial of the subspace defined in part (b) is $q=\mu q_0$ where
 \bea
\mu=(x_1+x_4+x_M)(x_2+x_5+x_M)\prod_{i=1}^{M-3}
(x_i+x_{i+2}+x_M),
 \eea
see \eqref{eq:neparnoM}. Hence, this subspace will be universal if we can prove that $\braket{q}{p}\ne0$.

By using the expansion \eqref{eq:Pol-p}, we deduce that $\braket{q}{p}=\braket{\mu q_0}{p}=\braket{\mu' q_0}{p}$ where $\mu'$ is obtained from $\mu$ by setting $x_M=0$. After cancelling $x_1 x_2\cdots x_{M-1}$, we obtain $\braket{q}{p}=\braket{q'}{p'}$, where $q'=\mu'q'_0$, $q'_0=\sum_f q'_f$ and the polynomial $p'$ is
defined by Eq. \eqref{eq:Pol-p} with $M$ replaced by $M-1$.
Hence $x_M$ has been eliminated, it occurs in neither $p'$
nor $q'$. For $M=11$, we were able to compute $\braket{q'}{p'}$.
\epf

We remark that in the case $M=6$ any subspace spanned by $5$
BSOVs is universal. However, in the case $M=7$ the 6 BSOVs that
we discard cannot be chosen arbitrarily, but there are several
other good choices such as
$(1,5,7),(2,5,7),(1,3,7),(2,4,7),(3,5,7),(4,6,7)$.

It should also be mentioned that, although it is not directly
applicable, our work may shed light on the understanding of the N-representability problem~\cite{Coleman},
where the single particle eigenvalues are invariants under LU.
For the $N=3,M=6$, let us consider a concrete example
of the universal subspace with dimension $5$, spanned by
\begin{equation}
e_{235}, e_{145}, e_{136}, e_{246}, e_{135},
\end{equation}
where $e_{ijk}=\ket{i}\we\ket{j}\we\ket{k}$.

In other words, any three-fermion pure-state with six single particle states is
unitary equivalent to
\begin{equation}
\label{eq:5dimVektor}
\ket{\psi}=a e_{235}+b e_{145}+c e_{136}+d e_{246}+z e_{135},
\end{equation}
and the coefficients can be chosen as
$a,b,c,d\ge0$, $z\in\bC$, and
$\|\psi\|^2=a^2+b^2+c^2+d^2+|z|^2=1$.
Without loss of generality, we can further assume that $a\ge b\ge c$.

The one particle reduced density matrix $\r_1$ of the state
$\r:=\proj{\psi}$, is given by (we choose the normalization $\tr\r_1=3$)
\[
\r_1=
\left(\begin{array}{cccccc}
b^2+c^2+|z|^2 & az & 0 & 0 & 0 & 0\\
az^* & a^2+d^2 & 0 & 0 & 0 & 0\\
0 & 0 & c^2+a^2+|z|^2 & bz & 0 & 0\\
0 & 0 & bz^* & b^2+d^2 & 0 & 0\\
0 & 0 & 0 & 0 & a^2+b^2+|z|^2 & cz \\
0 & 0 & 0 & 0 & cz^* & c^2+d^2
\end{array}\right).
\]
We have $\r_1=R_a\oplus R_b\oplus R_c$, a direct sum of three
$2\times2$ diagonal blocks. For $x=a,b,c$ let $D_x=\det R_x$.

Since $\tr R_x=1$ and $R_x\ge0$, we have $D_x\in[0,1/4]$ and the
eigenvalues of $R_x$ can be written as $\lambda_x$ and
$1-\lambda_x$ with $\lambda_x=(1+\sqrt{1-4D_x})/2\in[1/2,1]$.
Let us denote the eigenvalues of $\r_1$ arranged in
decreasing order as
$\lambda_1\ge\lambda_2\ge\cdots\ge\lambda_5\ge\lambda_6$.
Then $\lambda_i$ and $\lambda_{7-i}$ are the eigenvalues of the
same block $R_x$ of $\r_1$. Thus the following result
is a direct corollary of Proposition~\ref{pp:BSOV}.
\begin{corollary}
If the $\lambda_i$s are arranged in decreasing order,
then $\lambda_i+\lambda_{7-i}=1$ for $i=1,2,3$.
\end{corollary}

This hence gives an alternative proof for the N-representability
equalities and inequalities in the case
$N=3,M=6$~\cite{BD72,Kly05,Beth07,Kly09}.

\section{Universal Subspaces for $N>3$}
\label{sec:USNgeq3}

In this section we consider the generalization to $N>3$.
Note that for $N$ even the total number of configurations is $2K\choose N$, and that the number of basic single occupancy states is
$2^N{K\choose N}$. Consider the ratio
$r(N,K)={2K\choose N}/2^N{K\choose N}$.
We know that for any finite $N,K$, we have $r>1$. If we fix $N$,
we will have $\lim_{K\rightarrow\infty}r(N,K)=1$.
This means that the single occupancy states have measure $1$ for large $K$, so one may hope that a similar result as for $N=3$
might hold for $N>3$ at least when $K$ is large. However we show that this is not the case.

\subsection{The SOV-subspace is not universal for $N>3$}

 \bpp
Let $M$ and $N$ be integers such that $M\ge2N\ge8$. Let $V$ be a
complex Hilbert space of dimension $M$. Then there exist
$\ket{\psi}\in\wedge^N V$ which are not equivalent to any SOV.
 \epp
 \bpf
Write $M=2K$ if $M$ is even and $M=2K+1$ if $M$ is odd.
We fix an o.n. basis $\{\ket{i}: 1\le i\le M\}$ of $V$.
Let $V_i=\lin\{\ket{2i-1},\ket{2i}\}$, $i=1,\ldots,K$, and let
$V_{K+1}$ be the 1-dimensional subspace spanned by $\ket{M}$ if
$M$ is odd. Then the SOV-subspace, $\cS$, of $\wedge^N V$ is
given by
\bea \label{eq:DefW}
\cS=\sum_{1\le i_1<\cdots<i_N\le M}
V_{i_1}\wedge V_{i_2}\wedge\cdots\wedge V_{i_N}.
\eea

Let $G:=\GL(V)$ and let $f:G\times\cS\to\wedge^N V$ be the map
sending $(A,\ket{\psi})\to A\cdot\ket{\psi}$. Let $H$ be the subgroup of $G$ which leaves invariant each of the subspaces
$V_i$. Since $\cS$ is $H$-invariant, we can form the algebraic homogeneous vector bundle $G*_H W$ with projection map
$p:G*_H \cS\to G/H$ (see \cite[section 4.8]{PopVin}). As the dimension of $G*_H \cS$ is
$D:=\dim G+\dim \cS-\dim H$, we have
\bea
 \label{ea:D=even}
D &=& 4K(K-1)+2^N\binom{K}{N},\quad (M~{\rm even}); \notag\\
D &=& 4K^2+2^N\binom{K}{N}+2^{N-1}\binom{K}{N-1},\quad (M~{\rm
odd}). \eea
Since $f$ factorizes through the canonical map $G\times
\cS\to G*_H \cS$, we have $\dim G\cdot \cS\le D$. One can verify that
$D<\binom{M}{N}$, see Lemma \ref{le:D}. Thus $G\cdot \cS$ must be
a proper subset of $\wedge^N V$.
 \epf

 \bl
 \label{le:D}
For the dimension $D$ defined by Eq. \eqref{ea:D=even}, we have
$D<\binom{M}{N}$.
 \el
 \bpf
We first prove the assertion for even $M$, i.e., that
\bea  \label{ea:equivalentD}
 f(M,N) := \binom{M}{N}-2^N\binom{K}{N}>4K(K-1)
 \eea
when $M=2K\ge2N\ge8$. We use induction on $N$. One can easily verify $f(M,4)>0$.
Suppose $f(M,N)>0$ holds for some $N$ and
that $M\ge2N+2$. The inductive step follows from the identity
 \bea
 (N+1)\left(f(M,N+1)-f(M,N)\right)=(M-2N-1)f(M,N)
 +N\cdot 2^N \binom{K}{N}.
 \eea
which is easy to verify. Indeed, it implies that
$f(M,N+1)>f(M,N)$.

Next we show the assertion for odd $M=2K+1\ge 2N$ and $N\ge4$. For
$N=4$, we have
 \bea
 D - \binom{M}{4}
 &=&
 4K^2+2^4\binom{K}{4}+2^{3}\binom{K}{3} - \binom{2K+1}{4}
 \notag\\
 &<&
 4K+\binom{2K}{4}+2^{3}\binom{K}{3} - \binom{2K+1}{4}
 =2K(3-K)<0,
 \eea
where the first inequality is from Eq. \eqref{ea:equivalentD}. For
$N>4$, we have
 \bea
 D - \binom{M}{N}
 &=&
 4K^2+2^N\binom{K}{N}+2^{N-1}\binom{K}{N-1} - \binom{2K+1}{N}
 \notag\\
 &<&
 4K^2+(\binom{2K}{N} - 4K(K-1)) + (\binom{2K}{N-1} - 4K(K-1)) - \binom{2K+1}{N}\notag \\
&=&4K(2-K)<0,
 \eea
where the first inequality is from Eq. \eqref{ea:equivalentD}. So
the assertion is true for odd $M$. This completes the proof.
 \epf

To conclude this subsection, we give the lower bound for the
dimension of universal subspaces.

\bpp \label{prp:UnivCond} Let $V$ be a complex Hilbert space of
dimension $M$ and $\{\ket{i}:i=1,2,\ldots,M\}$ an o.n. basis of $V$.
Let $\cW$ be a complex subspace of $\wedge^N V$ spanned by some
basis vectors $\ket{i_1\we i_2\we\cdots\we i_N}$,
$i_1<i_2<\cdots<i_N$. If $\dim \cW<\binom{M}{N}-\binom{M}{2}$ then
$\cW$ is not universal (for the diagonal action of $\Un(M)$).
 \epp
 \bpf
Let $\Un(M)$ be the unitary group of $V$, identified with the group
of unitary matrices of order $M$ by using the basis $\{\ket{i}\}$.
Let $f:\Un(M)\times \cW\to \wedge^N V$ be the restriction of the
action map $\Un(M)\times\wedge^N V\to\wedge^N V$. Since $\cW$ is
spanned by the basis vectors, it is invariant under the action of
the maximal torus $T$ of $\Un(M)$ consisting of the diagonal unitary
matrices. We can form the equivariant vector bundle $\Un(M)*_T \cW$.
The map $f$ factorizes through the canonical map $\Un(M)\times
\cW\to\Un(M)*_T \cW$, and so we have
 \bea
\dim(\Un(M)\cdot \cW) \le \dim(\Un(M)*_T \cW) = M(M-1)+2d,
 \eea
where $d$ is the complex dimension of $\cW$. As
$d<\binom{M}{N}-\binom{M}{2}$, we obtain that $\dim(\Un(M)\cdot
\cW)<2\binom{M}{N}=\dim_\bR(\wedge^N V)$. Hence, $\Un(M)\cdot \cW$
must be a proper subset of $\wedge^N V$, i.e., $\cW$ is not
universal.
 \epf

We do not know that whether the bound $\binom{M}{N}-\binom{M}{2}$ is
sharp for universal spaces of general $N>2,M$. For $N=3$ and $M$
even, the bound is sharp by Proposition \ref{pp:BSOV} (a).

\subsection{BCS states are not LU-equivalent to single occupancy states}

We construct a concrete example that is not LU-equivalent to any
SOV, i.e. the BCS states. The BCS states are named after Bardeen,
Cooper and Schrieffer for there 1957 theory using the states to
build microscopic theory of superconductivity~\cite{BCS57}. Here we
consider the finite dimensional version of the BCS state
$\ket{\ps_{N,M}}$, which is an $N$-fermion state, i.e.
$\ket{\ps_{N,M}}\in\wedge^N V$ with $\dim V=M$, where both $N$ and
$M$ are even (see e.g.~\cite{Yang1962}). The $M$ states are `paired'
in a sense: say as $\{1,2\},\{3,4\},\ldots,\{M-1,M\}$. In each pair,
the two states are either both `occupied' by a pair of fermions or
both `empty'. The BCS state is an equal weight superposition of all
those `paired' states. Intuitively, BCS states are `always double
occupancy', so that they should be among the `hardest' ones to be
transformed by LU to some single occupancy state.

We list a few BCS states:
 \bea
 \label{ea:BCS}
 \ket{\ps_{2,M}}
 &=&
 \sum^{K}_{i=1} \wetw{2i-1}{2i},
 \notag\\  \label{eq:psi4M}
 \ket{\ps_{4,M}}
 &=&
 \sum_{i_1<i_2} \wefo{2i_1-1}{2i_1}{2i_2-1}{2i_2},
 \\
 & \vdots &
 \notag\\
 \ket{\ps_{N,M}}
 &=&
 \sum_{i_1<i_2<\cdots <i_{N/2}}
 \wetw{2i_1-1}{2i_1}
 \we\cdots\we
 \wetw{2i_{N/2}-1}{2i_{N/2}},
 \eea
where the $i_k\in\{1,\ldots,K\}$.

We shall use exterior multiplication and partial inner products with the BCS states, e.g., $\ket{a}\we\ket{\ps_{N,M}}$ or
$\braket{a}{\ps_{N,M}}$ where $\ket{a}\in V$. To simplify this computation, let $U_1,\ldots,U_K$ be $2\times2$ unitary matrices with $\det U_i=1$. So $U=\op^K_{i=1} U_i$ is a $M\times M$ unitary matrix. One can verify that
 \bea
 \label{ea:UandD}
 \we^2 U\wetw{2i-1}{2i} &=& \wetw{2i-1}{2i},
 \eea
where $i=1,\ldots,K$. Hence we obtain the stabilizer formulas
 \bea
 \label{ea:invariantBCS}
 \we^N U \ket{\ps_{N,M}}
 =
 \ket{\ps_{N,M}}.
 \eea
Consequently, in the expressions $\ket{a}\we\ket{\ps_{N,M}}$ and $\braket{a}{\ps_{N,M}}$, after an LU transformation we may assume that $\ket{a}$ is a linear combination of the $\ket{2i-1}$, $i=1,\ldots,K$.

By Lemma \ref{le:N=2SOV}, if $N=2$ then any BCS state is
LU-equivalent to some SOV. In Theorem \ref{thm:BCSnotSOV} we will
show that this is not the case for $N\ge4$. For this purpose we give a general criterion of deciding whether a state is
LU-equivalent to an SOV. We shall denote by
$\r_{12}$ the bipartite reduced density matrix of $\r$.
 \bl
 \label{le:singleoccupancy}
Let $\ket{\ps}\in\we^N V$, $\dim V=M\ge 2N$, be an antisymmetric
state and let $\rho=\proj{\psi}$. Then $\ket{\ps}$ is LU-equivalent to an SOV if and only if there is an o. n. basis
$\ket{a_1},\ldots,\ket{a_M}$ of $V$ such that
$\r_{12}(\ket{a_{2i-1}}\wedge\ket{a_{2i}})=0$
for all $i$ with $2i\le M$.
 \el
 \bpf
{\em Sufficiency}. By replacing $\ket{\ps}$ with an
LU-equivalent state, we may assume that $\ket{a_i}=\ket{i}$ for all $i$. We can write
 \bea
 \label{ea:NM}
 \ket{\ps}
 =
 \sum_{i_1<\cdots<i_N}c_{i_1,\ldots,i_N}
\ket{i_1\we\cdots\we i_N}.
 \eea
If $i_1=2s-1$ and $i_2=2s$, the hypothesis implies that
$\wetw{2s-1}{2s}\in\ker\tr_{3,\cdots,N}
\proj{\ps}=\ker\r_{12}$. So the coefficient
$c_{i_1,\cdots,i_N}=0$.
Similarly, $c_{i_1,\cdots,i_N}=0$ if
$\{2s-1,2s\}\subseteq\{i_1,\ldots,i_N\}$ for some $s$.
So $\ket{\ps}$ is a single occupancy state.

\textit{Necessity}. Suppose $\ket{\ps}$ is LU-equivalent to a single occupancy state $\ket{\ph}$, i.e.,
$\ket{\ps}=\we^N U \ket{\ph}$ with a unitary $U$. By definition of SOV we have $\bra{\ph}(\wetw{2s-1}{2s})=0$. This is equivalent to $(\bra{2s-1}\we\bra{2s})\s(\wetw{2s-1}{2s})=0$, where
$\s=\ketbra{\ph}{\ph}$. By tracing out all but the first two systems, we obtain
$(\bra{2s-1}\we\bra{2s})\s_{12}(\wetw{2s-1}{2s})=0$. It follows that
$\wetw{2s-1}{2s}\in\ker\s_{12}$. Thus, the assertion holds with
$\ket{a_i}=U\ket{i}$ for all $i$. This completes the proof.
 \epf

Note that the proof can be easily extended to the case in which LU is replaced by the diagonal action $\we^N V$ where $V$ is
invertible. Now we prove the main result on BCS states.

  \bt
 \label{thm:BCSnotSOV}
Let $M=2K$ and $N$ be even integers with $K\ge N\ge4$. Then the
BCS state $\ket{\ps_{N,M}}$ is not LU-equivalent to any SOV.
 \et
 \bpf
We start with the case $N=4$. Suppose $\ket{\ps_{4,M}}$ is
LU-equivalent to an SOV. By Lemma \ref{le:singleoccupancy}, there
is a nonzero decomposable 2-vector $\ket{\ph}=\wetw{a}{b}$ such
that the partial inner product $\braket{\ph}{\ps_{4,M}}=0$. By
the simplification mentioned beneath Eq. \eqref{ea:invariantBCS}, we may assume that
$\ket{a}=\sum^K_{j=1}a_{2j-1}\ket{2j-1}$, $a_1=1$, and
$\ket{b}=\sum^{M}_{j=2}b_j\ket{j}$.
As a special case of formula \eqref{eq:PartInnPr} we have
 \bea \label{eq:SpecFormula}
\braket{a\wedge b}{v_1\we v_2\we v_3\we v_4}=
\sum_{1\le i<j\le4} (-1)^{i+j-1}
\braket{a\wedge b}{v_i\we v_j}\ket{v_k\we v_l},
 \eea
where $\{i,j,k,l\}=\{1,2,3,4\}$ and $k<l$.
Recall from subsection \ref{se:UnivSO} that
$\overline{i}=\{i,i-1\}$ when $i$ is even and
$\overline{i}=\{i,i+1\}$when $i$ is odd. For convenience, we
shall write $\overline{i}<\overline{j}$ if $i<j$ and
$\overline{i}\ne\overline{j}$.
By using the formulae in Eqs. \eqref{eq:SpecFormula}
and \eqref{eq:psi4M}, we conclude that
 \bea \label{eq:Uslov}
\braket{a\wedge b}{i\we j}=0 \quad \text{if}\quad
\overline{i}<\overline{j},
 \eea
where $i,j\in\{1,\ldots,M\}$.
By setting $i=1$ in Eq. \eqref{eq:Uslov}, we conclude that
$b_j=0$ for $j>2$ and so $\ket{b}\propto\ket{2}$. Next, by
setting $i=2$ in Eq. \eqref{eq:Uslov}, we conclude that also
$a_j=0$ for $j>1$ and so $\ket{a\we b}\propto\ket{1\we 2}$.
As $\braket{1\we 2}{\ps_{4,M}}\ne0$, we have a contradiction.
Thus $\ket{\ps_{4,M}}$ is not LU-equivalent to an SOV.

Next, we use the induction on $N$ to show that there
is no nonzero decomposable 2-vector $\ket{\ph}=\ket{a\wedge b}$
such that the partial inner product $\braket{\ph}{\ps_{N,M}}=0$. By Lemma \ref{le:singleoccupancy}, this implies that the BCS state $\ket{\ps_{N,M}}$ is not LU-equivalent to any SOV.

Since we have already verified the first case, $N=4$, let us
assume that the assertion holds for $\ket{\ps_{N-2,M}}$ with
$N\ge6$ and $M\ge2(N-2)$. Suppose there is a nonzero decomposable 2-vector $\ket{\ph}=\ket{a\wedge b}$ such that
$\braket{\ph}{\ps_{N,M}}=0$ with $M\ge2N$. Let
 \bea
 \ket{a} &=& \sum^{M}_{j=1}a_j\ket{j} := \ket{\a} +
 \sum^{M}_{j=M-1}a_j\ket{j},
 \\
 \ket{b} &=& \sum^{M}_{j=1}b_j\ket{j} := \ket{\b} +
 \sum^{M}_{j=M-1}b_j\ket{j}.
 \eea
Since $\braket{\ph}{\ps_{N,M}}=0$, we have
 \bea
 0&=&
 \bra{M-1}\we\bra{M}(\braket{\ph}{\ps_{N,M}})
 =
 \bra{\ph}\bigg( (\bra{M-1}\we\bra{M})\ket{\ps_{N,M}} \bigg)
 \notag\\
 &=&
 \braket{\ph}{\ps_{N-2,M-2}}
 =\braket{\a\we\b}{\ps_{N-2,M-2}}.
 \eea
As $M-2>2(N-2)$, the induction hypothesis implies that
$\ket{\a\wedge\b}=0$. Thus $\ket{\a}$ and $\ket{\b}$ are linearly dependent. Consequently, we may assume that $\ket{\b}=0$, and moreover that $\ket{b}=\ket{M}$ and that all $a_{2i}=0$. Then by expanding the left hand side of
$\braket{a\wedge b}{\psi_{N,M}}=0$, we obtain a contradiction similarly as in case $N=4$. So there is no nonzero decomposable 2-vector $\ket{\ph}$ such that $\braket{\ph}{\ps_{N,M}}=0$. This completes the proof by induction.
 \epf

\section{Summary and Discussion}
\label{sec:Sum}

We have discussed universal subspaces for local unitary groups of fermionic systems. We have shown that for $N=3$, the SOV-subspace is universal. Furthermore, for $M$ even, we can always find a universal subspace, contained in the SOV-subspace, whose
dimension is equal to the lower bound $M(M-1)(M-5)/6$. Although our main tool is a natural application of Theorem $4.2$ in~\cite{ad10}, which can be used in small dimensions to construct universal subspaces by computers, the analytical proof we obtained for the general case is far from trivial. In fact, some special features of polynomials are used, which do not generalize to settle the odd $M$ case.

We have also shown that, for $N>3$, not all fermionic states are
LU-equivalent to a single occupancy state. Our argument is based on
dimension counting. For $M$ even, we give BCS states as concrete
examples that are not LU equivalent to any single occupancy state.
This is intuitive as BCS states are always paired so they are most
unlikely to be transformed with LU to something unpaired. Given that
for fixed $N$ almost all states are single occupancy in the large
$M$ limit, the BCS states are also among the `measure zero' states.

We wish our results shed light on further study of entanglement
properties on fermionic system, as well as other properties of
fermionic states such as the N-representability problem (as we
discussed for the $N=3,M=6$ case). We also leave some open
questions. One of them is to understand the achievability of the
dimension lower bound for $N=3$ and $M$ odd, or for even the case of
general $N>2,M$, which is worth further investigation.

\section*{Acknowledgments}

We thank Markus Grassl, Zhengfeng Ji, and Mary Beth Rusaki
for helpful discussion on this paper. LC was
mainly supported by MITACS and NSERC. The CQT is funded by the
Singapore MoE and the NRF as part of the Research Centres of
Excellence programme.
JC is supported by NSERC.
DD was supported in part by an NSERC Discovery Grant. His
computations using Singular were made possible by the
facilities of the Shared Hierarchical Academic Research Computing
Network (SHARCNET).
BZ is supported by NSERC and CIFAR.

\end{document}